\documentclass[submission,copyright,creativecommons]{eptcs}
\providecommand{\event}{Gandalf 2014}

\usepackage{amsmath}
\usepackage{amssymb}
\usepackage{amsthm}
\usepackage[utf8]{inputenc}
\usepackage[T1]{fontenc}
\usepackage{graphicx}
\usepackage{cite}
\usepackage{latexsym}
\usepackage{float}
\usepackage{caption}

\newtheoremstyle{mytheoremstyle} 
    {\topsep}                    
    {\topsep}                    
    {}                           
    {}                           
    {\scshape}                   
    {}                          
    {.5em}                       
    {}  

\theoremstyle{mytheoremstyle}
\bibliography{hg.bib}

\newtheorem{theorem}{Theorem}
\newtheorem{lemma}[theorem]{Lemma}
\newtheorem{definition}[theorem]{Definition}

\providecommand{\urlalt}[2]{\href{#1}{#2}}
\providecommand{\doi}[1]{doi:\urlalt{http://dx.doi.org/#1}{#1}}

\title{Hourglass Automata}
\author{
Yuki Osada, Tim French, Mark Reynolds, and Harry Smallbone \institute{The University of Western Australia.} \email{yuki.osada@research.uwa.edu.au, \{tim.french,mark.reynolds\}@uwa.edu.au, 21306592@student.uwa.edu.au}
}
\def\titlerunning{Hourglass Automata}
\def\authorrunning{Osada et al}

\begin{document}
\maketitle

\begin{abstract}
In this paper, we define the class of {\it hourglass automata}, which are timed automata with bounded clocks that can be made to progress backwards as well as forwards at a constant rate. We then introduce a new clock update for timed automata that allows hourglass automata to be expressed. This allows us to show that language emptiness remains decidable with this update when the number of clocks is two or less. This is done by showing that we can construct a finite untimed graph using clock regions from any timed automaton that use this new update.
\end{abstract}


\section{Introduction}
Hybrid systems, as defined by Alur et al.~\cite{achhhnossy95}, consist of finite automata where the transitions between states are guarded by a set of variables that change over time. Timed automata are a subset of hybrid systems that have been studied extensively due to their usefulness in industrial modelling. Timed automata as shown in Alur and Dill~\cite{ad94} model time using a finite set of monotonically increasing real-valued clocks. The clocks can be used in guards in automaton transitions and have a number of operations associated with them. They can be reset to 0 independently of each other. Cassez and Larsen~\cite{cl00} also define a class of timed automata, stopwatch automata, which allows the clocks to be stopped and started, resulting in an expressive power equal to linear hybrid automata. Stopping time generally causes undecidability, but Bérard et al.~\cite{bhs12} show it is possible to introduce a restricted form of stopping time with the interrupt timed automata, where clock levels are used to allow values that a clock can be compared to be determined during execution.

Alur and Dill prove the emptiness problem for timed automata with integer clock guards to be decidable. By creating a clock region automaton, state reachability is PSPACE-complete, leading to formal method applications such as UPPAAL~\cite{lp97}. UPPAAL is a model-checking tool that uses symbolic representations of integer clocks to verify timed automata models.

The decidability of language emptiness for several types of timed automata updates have been looked at Bouyer et al.~\cite{bdfp04}. They show that the updates $x := c$, $x := y$ are decidable, and $x := x + 1$, $x := y + c$ are decidable for 3 or more clocks only when no diagonal constraints are used. This shows that we can change decidability results by removing diagonal constraints.
Further, Brihaye et al.~\cite{bdgorw11} show that under bounded time horizons, reachability of linear hybrid automata can be improved.

We introduce a new update operation to the clocks $x := c_x - x$ over the bounded time range $[0, c_x]$, which allow clocks to simulate the behaviour of going backwards while actually moving forwards within a bounded time range. The motivation for this operation can be explained simply in terms of the hourglass problem. 


The hourglass problem~\cite{hg_problem} describes a situation of some finite set of hourglasses which must be used in conjunction to measure a period T. The clocks have co-prime maximum times with T, and can increase or decrease in a range $[0, c]$. Once a clock reaches the bounds of its range, it stops progressing until a flip is made. We define a flip to be a change in the clock's rate of change from increasing to decreasing or vice versa. Additionally, clocks may only have comparisons to its range bounds. To solve the problem we have two optimisation concerns: firstly, the total amount of time passed in the system, and secondly, the number of flips used to measure T.

Consider an example of timing an egg. The egg must be boiled for exactly 15 minutes with only a 7- and an 11-hourglass timers. This hourglass problem has a number of different solutions, including for example in the two clock problem solving the Bezout identity $ax - by=1$, where 1 is the greatest common divisor of the co-prime clocks. Applying this solution to the example gives $11x-7y=1$, with solution $x=2$, $y=3$ as the number of flips for each hourglass, leading to a total of 36 minutes with 5 flips to boil the egg.

The hourglass problem is partially
reducible to a timed automaton. By modelling the problem using UPPAAL without the ability to flip the clocks, we obtain a solution of 22 minutes. In this solution, the 7 and 11 clock are both started at $t=0$: at 7 minutes, the egg is set to boiling with 4 minutes remaining on the 11-clock. The 11-clock is reset at $t=11$ and the egg continues to be boiled for the following 11 minutes to give $4 + 11 = 15$.
However, there is a solution that is even faster than those that can be discovered using timed automata. The optimum solution starts with the egg boiling at $t=0$ for both clocks. At $t=7$, the 7-clock is reset and times 4 minutes until $t=11$. Now we use the new operation to flip the hourglass such that it times from 4 to 0, giving a total time of 15 minutes with 2 hourglass flips. This solution cannot be found using monotonically increasing clocks.
Introducing the new operation has useful applications in UPPAAL and, by extension, industrial model-checking.


\section{Applications}
The hourglass analogy is interesting from a problem solving aspect, but it is also very relevant in the context of industrial applications of hybrid automata. There are some key features of hourglass automata to characterize the problems we are interested in:
\begin{itemize}
\item Hourglasses can represent time only up to a fixed maximum time. Once the sand runs out of the hourglass the state of the clock cannot change\footnote{This is reminiscent of the normalisation process used in \cite{ad94}.}. 
\item Hourglasses cannot be compared directly to one another. Typically there is no way to tell whether one hourglass has more time remaining than another, (other than course estimations based on appearance).
\item Hourglasses can be flipped. Flipping an hourglass results in time running backwards, or an accumulated resource running out.
\end{itemize}
These properties can be found in numerous modelling tasks, particularly those associated with material flow networks. One can imagine that an hourglass represents a silo storing some commodity. A complex process can allow that commodity to accumulate for a time and then to be consumed. Direct applications of this include transportation network with various accumulation points (such as stockpiles or silos at a port). Silos have a fixed maximum, their current content cannot be easily determined, and the time to empty a silo is proportional to (or at least related to) the time taken to fill the silo. We are also interested in modelling more complex systems such as trucks moving minerals about a mine site. Trucks can act as mobile silos, but also a fragment of the process, such as the trucks fuel tank, is analogous to an hourglass. In many of these cases, we would expect the time to fill a silo, or truck, or fuel tank, to be proportional (though not equal) to the time to empty it. In this paper, we only consider clocks which allow time to change at rates of $1$ or $-1$ (i.e.\ silos fill up and empty at the same rate) to establish the theory. In future work we will investigate systems without this restriction. This extension would be similar to the multi-rate automata considered in \cite{henzinger}.

\section{Hourglass Automata}

\subsection{Preliminaries}
Timed automata are finite automata coupled with a finite set of real-valued variables, called {\it clocks}, which all increase at unit rate at all locations.
Each of these clocks may be reset to zero immediately after a transition, and these transitions can have clock constraints associated to them.
The clock constraints on transitions are called {\it guards}, and unless the constraints in the guard are met, the transition cannot be taken.
Similarly, constraints on locations are called {\it invariants}, and the constraint must be met to stay in that location.
The location invariants and transition guards may only consist of comparisons between clock values and non-negative integer constants.
We restrict it to integer values since any rational constant can be multiplied by a large enough value that they all become integers.

A {\it clock valuation} is a map from clock variables to non-negative real-values.
States in timed automata are written as location and clock valuation pairs, $(s, v)$.
Since there can be an infinite number of clock valuations, we can have an infinite number of states in timed automata.
There are two types of state transitions in a timed automaton:
\begin{itemize}
\item a {\it delay} transition $(s, v) \to_{d} (s, v + t)$ for some $t \ge 0$, where $v + t$ is the clock valuation $v$ with all clock values incremented by $t$.
\item an {\it action} transition $(s, v) \to_{a} (s', v')$ for some action $a$, where $s'$ is a location directly connected to $s$ in the automaton and $v'$ is the valuation $v$ after the transition updates are applied.
\end{itemize}
The only {\it clock update} available in the standard timed automata is the reset operation, $x := 0$.

For every clock $x \in X$, we will let $c_x$ be the largest integer constant that $x$ is compared against in the clock constraints.
Once $x$ exceeds $c_x$, the valuation of $x$ can be considered to be $\infty$ as its value would be indistinguishable from any value above $c_x$.
While the values are still distinguishable, we can split the fractional and integral components as comparisons in clock constraints are all with integers.
So for some $t \in \mathbb{R}^+$, let $fr(t)$ be the fractional component of $t$, and $\lfloor t \rfloor$ be the integral component of $t$.
This gives us $t = \lfloor t \rfloor + fr(t)$.

Hourglass automata are extensions of timed automata that can have clocks that go backwards as well as forwards, and these clocks are bound to the range $[0, c_x]$.
Additionally, the guards and invariants in a hourglass automata are limited to comparing a clock $x$ to its bounds: $0$ and $c_x$.
Comparing clocks or comparing clocks to other constants does not make sense for hourglasses.
Additionally, clocks in an hourglass automata can potentially be stopped to simulate the process of placing an hourglass on its side.

\subsection{Syntax}
We first introduce clocks that behave like hourglasses with clocks that are bound to a range.
\begin{definition}
An {\it hourglass clock} $x$ is a clock that can have a value in $[0, c_x]$, and rate of change in $\{-1, 0, 1\}$.
The value of a clock cannot exceed $c_x$ nor go below $0$, so the clock stops progressing past those points.
\end{definition}

We will now define the main feature of hourglasses, which is their ability to be flipped, and have the clock progress backwards.
\begin{definition}
A {\it flip} of a clock is an operation that multiplies the rate of change of the real-value of the clock by $-1$. All clocks initially have a rate of change of 1, so the flip operation alternates the rate of change between $1$ and $-1$.
\end{definition}

Another feature of hourglasses is that time can be stopped by placing them on their side.
\begin{definition}
Hourglass clocks can be stopped, and then started again. This simulates the process of placing a hourglass on its side, and later placing them back upright again. We will call the operation that switches the clock's state between these two as {\it toggling} the clock.
\end{definition}

The rate of change of a clock can be tracked by extending the definition of a timed automata state to include $d \in D : X \to \{-1, 1\}$ that maps clocks $x \in X$ to a value in the set $\{-1, 1\}$, which encodes the direction time is progressing for a clock. This multiples the number of states by $2^{|X|}$, where $|X|$ is the number of clocks in the system. We can further extend this to record the stopped clocks by mapping clocks to a value in $\{-1, -0, 0, 1\}$, resulting in a $4^{|X|}$ multiplier. The $-0$ is needed to store the direction time should progress when toggled back on. Additionally, flipping a clock when stopped, should swap the direction that time progresses in when the clock is started again.

With the above two definitions, we can now define our hourglass automata.
\begin{definition}
An {\it hourglass automaton} is a 7-tuple $A = (\sum, S, S_i, S_f, X, I, T)$ such that:
\begin{itemize}
\item $\sum$ is a finite set of actions,
\item $S$ is a finite set of locations,
\item $S_i \subset S$ is a set of initial locations,
\item $S_f \subset S$ is a set of final locations,
\item $X$ is a finite set of hourglass clocks,
\item $I : S \to C(X)$ is a mapping from locations to clock constraints (the {\it location invariants}), where a clock constraint $\phi \in C(X)$ maps a clock $x \in X$ to a constraint such that $\phi(x) = v(x) \prec c$ or $\phi(x) = c \prec v(x)$ or $\phi = \phi_1 \land \phi_2$, where $, c \in \{0, c_x\}, \prec \in \{<, \leq\}, \phi_1, \phi_2 \in C(X)$, and $v : X \to \mathbb{R}$ maps clocks to their clock valuation.
\item $T \subseteq S \times \sum \times C(X) \times 2^X \times 2^X \times S$ is a set of {\it transitions}, where the 6-tuple $\langle s, a, \phi, \mu_{flip}, \mu_{toggle}, s' \rangle$ is a transition from location $s$ to location $s'$ with the label $a$. This transition is enabled when the constraint $\phi$ is met, and taking this transition will flip a set of clocks $\mu_{flip} \subseteq X$, and toggle the progress of time in a set of clocks $\mu_{toggle} \subseteq X$. Note that the standard clock reset $x := 0$ is not available in hourglass automata.
\end{itemize}
\end{definition}

\subsection{Timed and Untimed Languages}
The \textit{initial states} of an automaton are the elements of $S_i \times v_0$, where $v_0(x) = 0$ for all $x \in X$.
The \textit{final states} of an automaton are states that are in the \textit{final locations}.
A run of a timed automaton is a sequence of {\it timed moves}, which represents a delay transition followed by an action transition, from an initial state.
A finite run is said to be {\it accepting} if the run ends in a final state.
For every finite run, we have a finite {\it timed word} that is a sequence of timed moves that represent the run.
We then say that the {\it timed language} accepted by a timed automaton is the set of timed words that are accepting runs.

The {\it language emptiness} problem is the problem of determining whether the language accepted by an automaton is empty or not.
If there is a finite number of states, then this problem is reduced to a graph problem, where we check whether there exists a path from an initial state node to a final state node.
This makes the problem decidable.
In timed automata, there can be an infinite number of states of the form $(s, v) \in S \times V$, where $|S|$ is finite but $|V|$ can be infinite.
Alur and Dill~\cite{ad90,ad94} proved the decidability of the language emptiness problem on timed automata by showing that they can construct a B\"uchi automaton from a timed automaton, and this B\"uchi automaton accepts exactly the set of {\it untimed words} that are equivalent to the timed words accepted by the timed automaton.
Instead of exact clock valuations, these untimed words use equivalence classes of clock valuations, called {\it clock regions}.

\section{Hourglass-Flip Expressible Update for Timed Automata}
\label{sec:flip}
We showed earlier that the states in a hourglass timed automaton can be represented by the addition of an extra map $d$ to represent each clock's direction (and whether they are stopped or not), so we now introduce a new clock update that allows us to keep a positive time progression like in the standard timed automata.
\begin{lemma}
The update $x := c - x$, where $x \in X$ and $c \in \mathbb{Z}_{\le c_x}$, with the reset operation $x := 0$ on timed automata is capable of expressing the flip operation with bounded clocks from the hourglass automata.
\end{lemma}
\begin{proof}
The hourglass automata only allows comparisons of clocks to the constants 0 and $c_x$, so we only care about how much time there is until the clock value reaches these two end points.\\
Let $0 \le x \le c_x$, then $a_1 = c_x - x$ is the time left before $x = c_x$ when time is moving forwards, and $b_1 = x$ is the time left before $x = 0$ when time is moving backwards.

If we now apply the update $x := c_x - x$, let $a_2$, $b_2$ be $a_1$, $b_1$ after we apply this update:\\
$a_2 = c_x - (c_x - x), b_2 = (c_x - x)$\\
$\implies a_2 = x$, $b_2 = c_x - x$\\
$\implies a_2 = b_1$, $b_2 = a_1$\\
$\therefore$ this update swaps the time until an end point is reached, which is exactly what the flip operation does in the hourglass automata.
Also, we remain in the $0 \le x \le c_x$ bound as expected ($c_x - x = c_x - [0, c_x] = [0, c_x]$).

This means we can let $c_x$ be our only bound, and restrict time to progressing forwards.
We also need to update the $d$ map when we flip a clock $x$: $d(x) := -d(x)$.

We do not compare the clock value to anything above $c_x$, so we can just let the clock become greater than $c_x$.
When a clock $x$ is at its bounds and stopped progressing, the expected result of a flip on $x$ is to let it start moving again towards the other bound, which would be $c_x$ away.
We can simulate this with a clock reset $x := 0$ and state change $d(x) := -d(x)$, then let the clock progress towards $c_x$.
For this, we just have to check if $v(x) \ge c_x$, and execute the above steps in place of the $x := c_x - x$ update.

$\therefore$ We can represent the hourglass automata's bounded clock and flip operation with timed automata, using the standard clock reset operation and the new update operation.
\end{proof}

We will use timed automata extended with the $x := c_x - x$ update to show that the language emptiness problem on hourglass automata is decidable for two clocks or less.
To prove decidability of the language emptiness problem on the hourglass automata, we will construct a finite untimed automaton that accepts untimed words that are equivalent to the timed words that are accepted by the corresponding timed automaton using the $x := c_x - x$ update.
A finite untimed automata would have a finite number of states, so language emptiness would be decidable.

\subsection{Clock Regions}
In the standard timed automata, a finite untimed automata can be constructed using clock regions, which represent sets of clock valuations.
From the reasoning given by Alur, Courcoubetis, and Dill~\cite{acd90}, we can assume all comparisons will be between clocks and integers.
Thus for every clock $x \in X$, we only care about its current integer value.
The ordering of the fractional components is the other piece of information that is important so we know the order in which the integral components increase.
This allowed the partition of the clock valuation space into equivalence classes.
We will use a similar approach.

The original three constraints for two clock valuations, $v$ and $v'$, to be equivalent ($v \cong v'$) are:
\begin{enumerate}
\item For all $x \in X$ either $v(x) \ge c_x$ and $v'(x) \ge c_x$, or $\lfloor v(x) \rfloor = \lfloor v'(x) \rfloor$
\item For all $x, y \in X$ such that $v(x) \le c_x$ and $v(y) \le c_y$,\\
$fr(v(x)) \le fr(v(y))$ if and only if $fr(v'(x)) \le fr(v'(y))$
\item For all $x \in X$ such that $v(x) \le c_x$, $fr(v(x)) = 0$ if and only if $fr(v'(x)) = 0$
\end{enumerate}
With the new update, $x := c_x - x$, we add one more constraint:
\begin{enumerate}
\item[4.] For all $x, y \in X$ such that $v(x) \le c_x$ and $v(y) \le c_y$,\\
$fr(v(x)) + fr(v(y)) \le 1$ if and only if $fr(v'(x)) + fr(v'(y)) \le 1$ and\\
$fr(v(x)) + fr(v(y)) \ge 1$ if and only if $fr(v'(x)) + fr(v'(y)) \ge 1$.
\end{enumerate}

\noindent
\begin{lemma}
\label{lemma:consistentupdate}
The fractional component constraints in 2 is mapped to the fractional component constraints in 4 and vice versa when the update $v'(x) := c_x - v(x)$ is made and the valuation of the clock isn't an integer. This means constraint 4 is required to preserve information when the flip operation is made.
\end{lemma}
\begin{proof}
From constraint 2 to constraint 4:

Let $x, y \in X$, $v(x) \le c_x$, $v(x) \notin \mathbb{Z}$, $v(y) \le c_y$, and $fr(v(x)) \le fr(v(y))$.\\
Applying the update $v'(x) := c_x - v(x)$ gives us:\\
$fr(v'(x)) \le fr(v(y))$
$\iff fr(c_x - v(x)) \le fr(v(y))$
$\iff 1 - fr(v(x) - c_x) \le fr(v(y)) \iff$\\
$1 - fr(v(x)) \le fr(v(y))$
$\iff 1 \le fr(v(x)) + fr(v(y))$
$\iff fr(v(x)) + fr(v(y)) \ge 1$\\
$\therefore fr(v(x)) \le fr(v(y)) \to fr(v(x)) + fr(v(y)) \ge 1$

Similarly, $fr(v(y)) \le fr(v(x))$ updates to $fr(v(x)) + fr(v(y)) \le 1$.\\

\noindent
From constraint 4 to constraint 2:

Let $x, y \in X$ such that $v(x) \le c_x$, $v(x) \notin \mathbb{Z}$, $v(y) \le c_y$, and $fr(v(x)) + fr(v(y)) \le 1$.\\
Applying the update $v'(x) := c_x - v(x)$ gives us:\\
$fr(v'(x)) + fr(v(y)) \le 1$
$\iff fr(c_x - v(x)) + fr(v(y)) \le 1$
$\iff fr(v(y)) \le 1 - fr(c_x - v(x)) \iff$\\
$fr(v(y)) \le fr(v(x) - c_x)$
$\iff fr(v(y)) \le fr(v(x))$\\
$\therefore fr(v(x)) + fr(v(y)) \le 1 \to fr(v(y)) \le fr(v(x))$


Similarly, $fr(v(x)) + fr(v(y)) \ge 1$ updates to $fr(v(x)) \le fr(v(y))$.\\

\noindent
Note that if $v(x) \in \mathbb{Z}$, then $fr(c_x - v(x)) = 0$, so the fractional constraints do not change and it is still consistent.
\end{proof}

With only the original three constraints, $\cong$ is known to be an equivalence relation.
\begin{lemma}
$\cong$ remains an equivalence relation with this new constraint.
\end{lemma}
\begin{proof}We show that the three properties of an equivalence relation are maintained.

\noindent
Reflexive:
This is trivially true.

\noindent
Symmetric:

Let $v, v' \in V$, and
$\forall x, y \in X$ such that $v(x) \le c_x$ and $v(y) \le c_y$.\\
Either $v(x) = c_x \land v'(x) = c_x$ or $\lfloor v(x) \rfloor = \lfloor v'(x) \rfloor$, so $v'(x) \le c_x$ (similarly for $y$).\\
$\therefore v \cong v' \implies v' \cong v$.

\noindent
Transitive:

Let $v, v', v'' \in V$, and:
$\forall x, y \in X$ such that $v(x) \le c_x$, $v'(x) \le c_x$, $v(y) \le c_y$ and $v'(y) \le c_y$.\\
Either $v(x) = c_x \land v'(x) = c_x \land v''(x) = c_x$ or $\lfloor v(x) \rfloor = \lfloor v'(x) \rfloor = \lfloor v''(x) \rfloor$, so $v''(x) \le c_x$ (similarly for $y$).
$\implies \forall x, y \in X$ such that $v(x) \le c_x$ and $v(y) \le c_y$\\
$f(v(x)) + fr(v(y)) \le 1 \iff f(v'(x)) + fr(v'(y)) \le 1 \iff f(v''(x)) + fr(v''(y)) \le 1$\\
$f(v(x)) + fr(v(y)) \ge 1 \iff f(v'(x)) + fr(v'(y)) \ge 1 \iff f(v''(x)) + fr(v''(y)) \ge 1$\\
$\therefore v \cong v' \land v' \cong v'' \implies v \cong v''$.

The other constraints are unaffected, so $\cong$ still defines an equivalence relation.
\end{proof}

This means we have a partitioning of all possible clock valuations.
Each equivalence class of $\cong$ is called a {\it region}.

\begin{lemma}
The number of regions induced by $\cong$ is bounded by:\\
$\prod_{x \in X}(2\cdot(c_x + 1)) \cdot \lvert X \rvert ! \cdot 2^{\lvert X \rvert - 1} \cdot (\lvert X \rvert + 1)^{|X|} \cdot (\lvert X \rvert + 1)^{|X|}$
\end{lemma}
\begin{proof}
A region can be described with five arrays:
\begin{enumerate}
\item For each clock $x \in X$, the interval in which $x$ lies. Possible options for some clock $x$ are:\\
$\{[0, 0], (0, 1), [1, 1], ..., (c_x-1, c_x), [c_x, c_x], (c_x, \infty)\}$\\
The number of possible ways to construct this array is $\prod_{x \in X}(2\cdot(c_x + 1))$.
\item A permutation of $X_{v \le c_x}$, $\beta : X \to \{1, 2, ..., \lvert X_{v \le c_x} \rvert\}$, giving the $\le$ ordering of the fractional components of the clocks with valuations that are still distinguishable.\\
The number of possible ways to construct this array is $\lvert X_{v \le c_x} \rvert !$.
\item A boolean array of whether the fractional component of a clock is equal to the fractional component of the succeeding clock in the above permutation.\\
The number of possible ways to construct this array is $2^{\lvert X_{v \le c_x} \rvert - 1}$.
\item For each clock $x \in X_{v \le c_x}$, the greatest integer $i \in \mathbb{N}_{\le \lvert X_{v \le c_x} \rvert}$ such that $fr(v(x)) + fr(v(x_i)) < 1$, where $\beta(x_i) = i$ or $v(x_0) = 0$ if no such $i$ exists.\\
The number of possible ways to construct this array is $(\lvert X_{v \le c_x} \rvert + 1)^{|X_{v \le c_x}|}$.
\item For each clock $x \in X_{v \le c_x}$, the greatest integer $i \in \mathbb{N}_{\le \lvert X_{v \le c_x} \rvert}$ such that $fr(v(x)) + fr(v(x_i)) \le 1$, where $\beta(x_i) = i$ or $v(x_0) = 0$ if no such $i$ exists.\\
The number of possible ways to construct this array is $(\lvert X_{v \le c_x} \rvert + 1)^{|X_{v \le c_x}|}$.
\end{enumerate}
\begin{lemma}
Every equivalence class of $\cong$ can be represented by some five-tuple $\langle \alpha, \beta, \gamma, \zeta, \eta \rangle$, containing the arrays above.
\end{lemma}
\begin{proof}
Let $v, v' \in V$.

$\alpha$ consists of intervals of the form $[a, a]$ or $(a, a+1)$ for $a \in \mathbb{N}_{< c_x}$, $[c_x, c_x]$, and $(c_x, \infty)$.
\begin{itemize}
\item $\alpha(x) = [a, a] \iff \lfloor v(x) \rfloor = a, \lfloor v'(x) \rfloor = a, fr(v(x)) = 0, fr(v'(x)) = 0$\\
$\iff \lfloor v(x) \rfloor = \lfloor v'(x) \rfloor$, $fr(v(x)) = 0 \land fr(v'(x)) = 0$
\item $\alpha(x) = (a, a) \iff \lfloor v(x) \rfloor = a, \lfloor v'(x) \rfloor = a, fr(v(x)) \neq 0, fr(v'(x)) \neq 0$\\
$\iff \lfloor v(x) \rfloor = \lfloor v'(x) \rfloor$, $fr(v(x)) \neq 0 \land fr(v'(x)) \neq 0$
\item $\alpha(x) = [c_x, c_x] \iff v(x) \ge c_x, v'(x) \ge c_x, fr(v(x)) = 0, fr(v'(x)) = 0$\\
$\iff v(x) \ge c_x \land v'(x) \ge c_x$, $fr(v(x)) = 0 \land fr(v'(x)) = 0$
\item $\alpha(x) = (c_x, \infty) \iff v(x) > c_x, v'(x) > c_x$\\
$\iff v(x) > c_x \land v'(x) > c_x$
\end{itemize}
$\therefore \alpha$ satisfies constraints 1 and 3.

$\beta$ and $\gamma$ give us an $\le$ ordering of the fractional components of the clocks
with valuations that are still distinguishable while giving us all sub-strings where the fractional components are equal.
\begin{itemize}
\item $\beta(x) = \beta(y)$\\
$\iff v(x) \le c_x \land v(y) \le c_y \land v'(x) \le c_x \land v'(y) \le c_y \land fr(v(x)) = fr(v(y)) \land fr(v'(x)) = fr(v'(y))$
\item $\beta(x) < \beta(y) \land (\land_{a \in \{a \in X | \beta(x) \le \beta(a) \le \beta(y)\}}(\gamma(a)))$\\
$\iff v(x) \le c_x \land v(y) \le c_y \land v'(x) \le c_x \land v'(y) \le c_y \land fr(v(x)) = fr(v(y)) \land fr(v'(x)) = fr(v'(y))$
\item $\beta(x) < \beta(y) \land \lnot (\land_{a \in \{a \in X | \beta(x) \le \beta(a) \le \beta(y)\}}(\gamma(a))) \iff fr(v(x)) < fr(v(y)), fr(v'(x)) < fr(v'(y))$\\
$\iff v(x) \le c_x \land v(y) \le c_y \land v'(x) \le c_x \land v'(y) \le c_y \land fr(v(x)) < fr(v(y)) \land fr(v'(x)) < fr(v'(y))$
\end{itemize}
$\therefore$ $\beta$ and $\gamma$ satisfy constraint 2.

$\zeta$ and $\eta$ with $\beta$ give us for each clock, the set of clocks that can be added so that sum is less than zero, equal to zero, and greater than zero.
\begin{itemize}
\item $\beta(y) \le \zeta(x) \iff v(x) \le c_x \land v'(x) \le c_x \land v(y) \le c_y \land v'(y) \le c_y \land fr(v(x)) + fr(v(y)) < 1$
\item $\eta(x) < \beta(y) \iff v(x) \le c_x \land v'(x) \le c_x \land v(y) \le c_y \land v'(y) \le c_y \land fr(v(x)) + fr(v(y)) > 1$
\item $\zeta(x) < \beta(y) \le \eta(x) \iff v(x) \le c_x \land v'(x) \le c_x \land v(y) \le c_y \land v'(y) \le c_y \land fr(v(x)) + fr(v(y)) = 1$
\end{itemize}
$\therefore$ $\zeta$ and $\eta$ with $\beta$ satisfy constraint 4.
\end{proof}

\noindent
$\therefore$ The number of clock regions is bounded by:\\
$\prod_{x \in X}(2\cdot(c_x + 1)) \cdot \lvert X \rvert ! \cdot 2^{\lvert X \rvert - 1} \cdot (\lvert X \rvert + 1)^{|X|} \cdot (\lvert X \rvert + 1)^{|X|}$
\end{proof}

We next look at the elapsing of time, clock constraints, and clock updates.

\begin{lemma}
Let $v_1$ and $v_2$ be two clock valuations, $t$ be a non-negative integer, $\phi$ be a clock constraint, and $\lambda, \mu \subseteq X$ be a set of clocks.
The following properties hold when the number of clocks is two or less:
\begin{enumerate}
\item $v_1 \cong v_2 \land t \in \mathbb{Z}^+ \implies v_1 + t \cong v_2 + t$
\item $v_1 \cong v_2 \implies \forall t_1 \in \mathbb{R}^+ \exists t_2 \in \mathbb{R}^+ [v_1 + t_1 \cong v_2 + t_2]$
\item $v_1 \cong v_2 \implies v_1$ satisfies $\phi \iff v_2$ satisfies $\phi$
\item $v_1 \cong v_2 \implies v_1[\lambda := 0] \cong v_2[\lambda := 0]$
\item $v_1 \cong v_2 \implies v_1[\mu := c_\mu - \mu] \cong v_2[\mu := c_\mu - \mu]$
\end{enumerate}
\end{lemma}

\begin{proof}(1)


\noindent
Let $v_1 \cong v_2, t \in \mathbb{Z}^+$, then:\\
$\forall x, y \in X$
$fr(v_1(x)) + fr(v_1(y)) \le 1 \iff fr(v_2(x)) + fr(v_2(y)) \le 1$ and\\
\hspace*{4em}$fr(v_1(x)) + fr(v_1(y)) \ge 1 \iff fr(v_2(x)) + fr(v_2(y)) \ge 1$\\
$\implies \forall x, y \in X$
$fr(v_1(x)+t) + fr(v_1(y)+t) \le 1 \iff fr(v_2(x)+t) + fr(v_2(y)+t) \le 1$ and\\
\hspace*{6em}$fr(v_1(x)+t) + fr(v_1(y)+t) \ge 1 \iff fr(v_2(x)+t) + fr(v_2(y)+t) \ge 1$\\
$\implies v_1 + t \cong v_2 + t$
\end{proof}

\begin{proof}(2)


\noindent
When $t_1 = 0$:
$t_2 = 0$

\noindent
When $t_1 \in \mathbb{Z}^+$:
$t_2 = t_1$ as above.

\noindent
When $0 < t_1 < 1$:
The existence of a $t_2$ is not guaranteed when there are more than two clocks (See~\ref{lemma:threeclockissue}).

\noindent
We limit it to two clocks $X = \{x, y\}$ here, and $v_1 \cong v_2$:
\begin{enumerate}
\item Either $v_1(x) \ge c_x$ and $v_2(x) \ge c_x$, or $\lfloor v_1(x) \rfloor = \lfloor v_2(x) \rfloor$ (similarly for $y$)
\item $v_1(x) \le c_x \land v_1(y) \le c_x \implies$
$fr(v_1(x)) \le fr(v_1(y)) \iff fr(v_2(x)) \le fr(v_2(y))$ and\\
\hspace*{11.7em}$fr(v_1(y)) \le fr(v_1(x)) \iff fr(v_2(y)) \le fr(v_2(x))$
\item $v_1(x) \le c_x \implies$
$(fr(v_1(x)) = 0 \iff fr(v_2(x)) = 0)$ (similarly for $y$)
\item $v_1(x) \le c_x \land v_1(y) \le c_y \implies$
$fr(v_1(x)) + fr(v_1(y)) \le 1 \iff fr(v_2(x)) + fr(v_2(y)) \le 1$ and\\
\hspace*{11.7em}$fr(v_1(x)) + fr(v_1(y)) \ge 1 \iff fr(v_2(x)) + fr(v_2(y)) \ge 1$
\end{enumerate}
\begin{itemize}
\item Case $v_1(x) + t_1 > c_x, v_1(y) + t_1 > c_y$:
$t_2 > c_x - v_1(x) \land t_2 > c_y - v_1(y) \implies t_2$ is not bounded from above, and
$t_2 = t_1 + 1$ will result in $v_2(x) + t_2 > c_x\land v_2(y) + t_2 > c_y$.
\item Case $v_1(x) + t_1 > c_x, v_1(y) + t_1 \le c_y$:
$c_x - v_2(x) < t_2 \land t_2 \le c_y - v_2(y)$\\
If $c_y - v_2(y) \ge 1$, then $t_2 = 1$ will work since $c_x - 1 < v_2(x) \le c_x$ must be true.\\
If $0 < c_y - v_2(y) < 1 \implies \lfloor v_2(y) \rfloor = c_y-1 \implies \lfloor v_1(y) \rfloor = c_y-1 \implies$\\
$fr(v_1(y)) < fr(v_1(x)) \implies fr(v_2(y)) < fr(v_2(x)) \implies$\\
$t_2 = 1 - fr(v_2(y))$ will work since $1-fr(v_2(x)) < t_2 \le 1 - fr(v_2(y))$ is the range $t_2$ must be in and this range is not empty.\\
$c_y - v_2(y) = 0$ cannot be true since that would mean $v_2(y) = c_y \iff v_1(1) = c_y \iff t_1 = 0$, which is a contradiction since $0 < t_1 < 1$.
\item Case $v_1(x) + t_1 \le c_x, v_1(y) + t_1 > c_y$: We can find a $t_2$ using the same reasoning as above.
\item Case $v_1(x) + t_1 \le c_x, v_1(y) + t_1 \le c_y, fr(v(y)) \le fr(v(x))$: (The $fr(v(x)) \le fr(v(y))$ case is similar)

\end{itemize}
Let's assume neither integral components change when $t_1$ is added, then constraints 1 to 3 only require that $t_2$ doesn't change either integral components as well.\\
The regions we get from the fourth constraint:
\begin{enumerate}
\item $fr(v(x)) + fr(v(y)) < 1 \land fr(v(y)) = 0$
\item $fr(v(x)) + fr(v(y)) < 1$
\item $fr(v(x)) + fr(v(y)) = 1 \land fr(v(x)) \neq 0$
\item $fr(v(x)) + fr(v(y)) > 1 \land fr(v(x)) \neq 0$
\end{enumerate}

\noindent
From any point in these regions, you can reach a point in\\the following region:
\begin{enumerate}
\item[$1\to2$] Let $t_2 = \epsilon$, where $0 < \epsilon < 1 - (fr(v(x)) + fr(v(y)))$
\item[$2\to3$] Let $t_2 = 1 - (fr(v(x)) + fr(v(y)))$
\item[$3\to4$] Let $t_2 = \epsilon$, where $\epsilon > 0$, $\lfloor v(x) \rfloor = \lfloor v(x)+\epsilon \rfloor$,\\$\lfloor v(y) \rfloor = \lfloor v(y)+\epsilon \rfloor$
\item[$4\to1$] Let $t_2 = min(1 - fr(v(x)), 1 - fr(v(y)))$
\begin{figure}[htb]
\vspace{-20em}
\hfill\begin{minipage}{.42\textwidth}\centering
\includegraphics{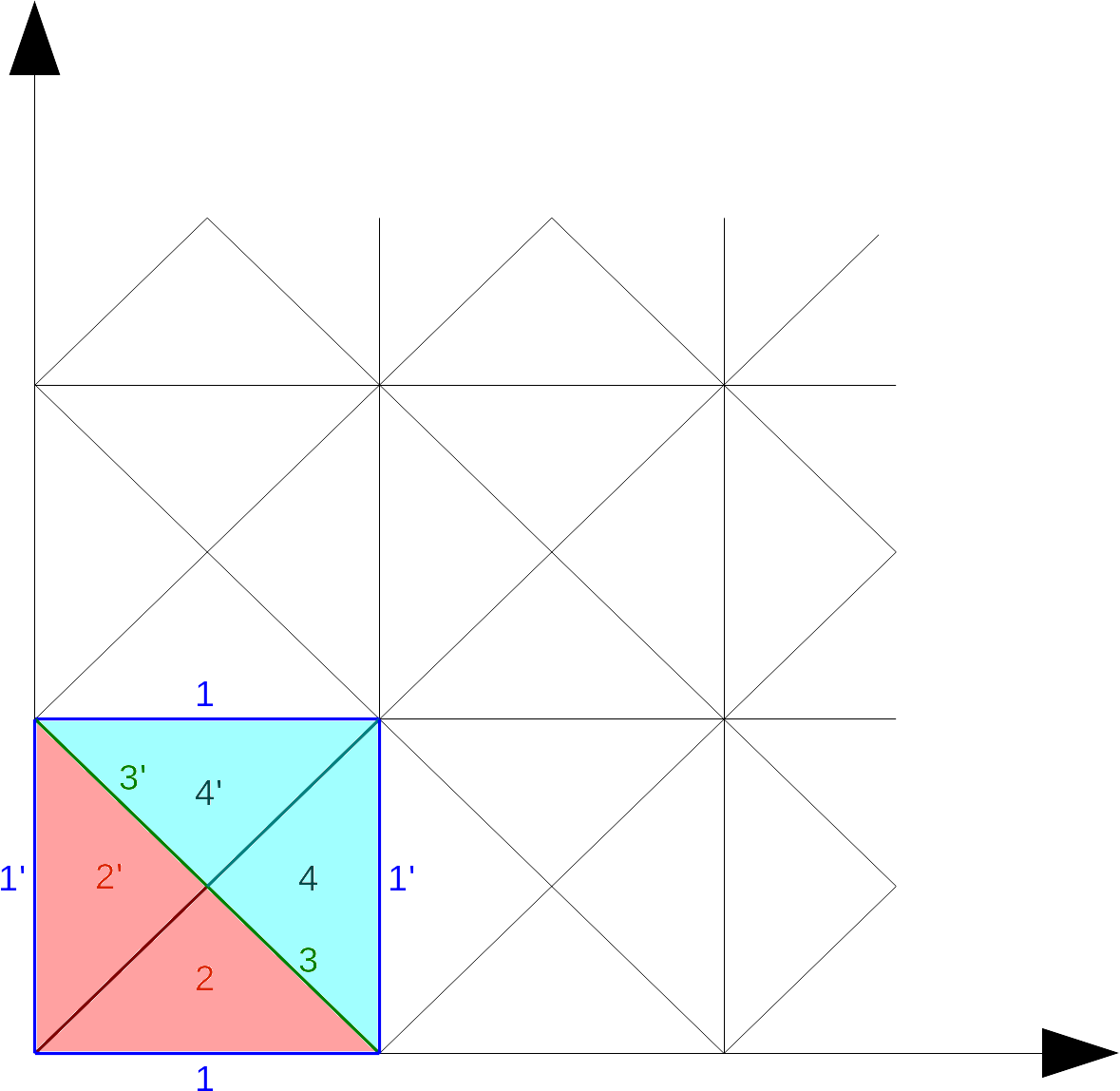}
\caption{\texttt{The four region types}}
\label{fig:regiontypes}
\end{minipage}
\end{figure}
\end{enumerate}

\noindent
After the fourth region, $fr(v(x)) = 0$ and the integral component of x is incremented before starting again at the first region, but with x and y variable positions swapped.
The transition from 4 to 1 also aligns with the integral component increment and fractional component ordering change, so they are taken into account.


\noindent
When $t_1 > 1$, $t_1 \notin \mathbb{Z}^+$:\\
Let $t_1' = t_1 - \lfloor t_1 \rfloor \implies 0 < t_1' < 1$, and\\
$\lfloor t_1 \rfloor \in \mathbb{Z}^+ \iff v_1 + \lfloor t_1 \rfloor \cong v_2 + \lfloor t_1 \rfloor$, so if we let $v_1' = v_1 + \lfloor t_1 \rfloor, v_2' = v_2 + \lfloor t_1 \rfloor$,\\
then we just need to find $t_2'$, where $v_1' + t_1' \cong v_2' + t_2'$.
\end{proof}

\begin{proof}(3) Let $x, y \in X$

\noindent
$v_1 \cong v_2$ and $v_1$ satisfies $\phi$, where $\phi$ compares clocks to integers:\\
$\forall x \in X, a_x, b_x \in \mathbb{Z}_{\le c_x}^+, a_x \le v(x) \le b_x$\\
$\implies a_x, b_x, a_y, b_y \in \mathbb{Z}_{\le c_x}^+, a_x \le v_1(x) \le b_x \land a_y \le v_1(y) \le b_y \land$\\
$\lfloor v_1(x) \rfloor = \lfloor v_2(x) \rfloor \land \lfloor v_1(y) \rfloor = \lfloor v_2(y) \rfloor \land$\\
$fr(v_1(x)) = 0 \iff fr(v_2(x)) = 0 \land fr(v_1(y)) = 0 \iff fr(v_2(y)) = 0$\\
$\implies a_x \le v_2(x) \le b_x \land a_y \le v_2(y) \le b_y$\\
$\implies v_2$ satisfies $\phi$, and it is unaffected by the new constraint.
\end{proof}

\begin{proof}(4) Let $x, y \in X, \lambda \subseteq X$

\noindent
$v_1 \cong v_2 \implies$\\
$fr(v_1(x)) + fr(v_1(y)) \le 1 \iff fr(v_2(x)) + fr(v_2(y)) \le 1 \land$\\
$fr(v_1(x)) + fr(v_1(y)) \ge 1 \iff fr(v_2(x)) + fr(v_2(y)) \ge 1$\\

\noindent
$v_1[\lambda := 0], (x \in \lambda \lor y \in \lambda) \implies fr(v_1(x)) + fr(v_1(y)) \le 1$\\
$v_2[\lambda := 0], (x \in \lambda \lor y \in \lambda) \implies fr(v_2(x)) + fr(v_2(y)) \le 1$\\
$\therefore v_1[\lambda := 0], v_2[\lambda := 0], (x \in \lambda \lor y \in \lambda) \implies fr(v_1(x)) + fr(v_1(y)) \le 1, fr(v_2(x)) + fr(v_2(y)) \le 1$\\
$\implies v_1[\lambda := 0] \cong v_2[\lambda := 0]$ since the other constraints are already known to be satisfied.
\end{proof}


\begin{proof}(5) Let $x \in \mu, \mu \subseteq X$

If $v(x) \ge c_x$, then the $flip$ operation is converted to a reset operation as explained in Section~\ref{sec:flip}, so we can assume $v(x) < c_x$.

\noindent
$v_1 \cong v_2 \implies \lfloor v_1(x) \rfloor = \lfloor v_2(x) \rfloor \land (fr(v_1(x)) = 0 \iff fr(v_2(x)) = 0)$

\noindent
$fr(v_1(x)) = 0 \implies \lfloor c_x - v_1(x) \rfloor = c_x - v_1(x) = c_x - \lfloor v_1(x) \rfloor, fr(c_x - v_1(x)) = 0$\\
$fr(v_2(x)) = 0 \implies \lfloor c_x - v_2(x) \rfloor = c_x - v_2(x) = c_x - \lfloor v_2(x) \rfloor, fr(c_x - v_2(x)) = 0$\\
$fr(v_1(x)) \ne 0 \implies \lfloor c_x - v_1(x) \rfloor = c_x - \lceil v_1(x) \rceil = c_x - (\lfloor v_1(x) \rfloor + 1), fr(c_x - v_1(x)) \ne 0$\\
$fr(v_2(x)) \ne 0 \implies \lfloor c_x - v_2(x) \rfloor = c_x - \lceil v_2(x) \rceil = c_x - (\lfloor v_2(x) \rfloor + 1), fr(c_x - v_2(x)) \ne 0$\\
$\therefore \lfloor c_x - v_1(x) \rfloor = \lfloor c_x - v_2(x) \rfloor \land (fr(c_x - v_1(x)) = 0 \iff fr(c_x - v_2(x)) = 0)$ is true, so constraints 1 and 3 are maintained after the flip operation.\\

\noindent
Let $x, y \in X, \mu \subseteq X$


\noindent
From Lemma~\ref{lemma:consistentupdate}, we know that if $x \in \mu, y \notin \mu, v(x) \notin \mathbb{Z}$ or $x, y \in \mu, v(x) \notin \mathbb{Z}, v(y) \in \mathbb{Z}$:\\
$fr(v(x)) \le fr(v(y)) \to fr(v(x)) + fr(v(y)) \ge 1$,\\
$fr(v(y)) \le fr(v(x)) \to fr(v(x)) + fr(v(y)) \le 1$,\\
$fr(v(x)) + fr(v(y)) \le 1 \to fr(v(y)) \le fr(v(x))$,\\
$fr(v(x)) + fr(v(y)) \ge 1 \to fr(v(x)) \le fr(v(y))$\\
and $x, y \in \mu, v(x), v(y) \notin \mathbb{Z} \implies$\\
$fr(v(x)) \le fr(v(y)) \to fr(v(x)) + fr(v(y)) \ge 1 \to fr(v(y)) + fr(v(x)) \ge 1 \to fr(v(y)) \le fr(v(x))$,\\
$fr(v(y)) \le fr(v(x)) \to fr(v(x)) + fr(v(y)) \le 1 \to fr(v(y)) + fr(v(x)) \le 1 \to fr(v(x)) \le fr(v(y))$,\\
$fr(v(x)) + fr(v(y)) \le 1 \to fr(v(y)) \le fr(v(x)) \to fr(v(x)) \le fr(v(y)) \to fr(v(x)) + fr(v(y)) \ge 1$,\\
$fr(v(x)) + fr(v(y)) \ge 1 \to fr(v(x)) \le fr(v(y)) \to fr(v(y)) \le fr(v(x)) \to fr(v(x)) + fr(v(y)) \le 1$\\
and $x \in \mu, v(x) \in \mathbb{Z}$ or $x, y \in \mu, v(x), v(y) \in \mathbb{Z} \implies$ no fractional component constraints change with the flip operation.

\noindent
From the above mappings, we can see that all fractional component constraints are being mapped consistently to a fractional component constraint.\\
$\therefore$ constraints 2 and 4 are maintained after the flip operation.
\end{proof}

\begin{lemma}
\label{lemma:threeclockissue}
The progression of time is where this clock update faces undecidability issues when there are more than two clocks. The condition $v_1 \cong v_2 \implies \forall t_1 \in \mathbb{R}^+ \exists t_2 \in \mathbb{R}^+ [v_1 + t_1 \cong v_2 + t_2]$ from above is not met with three clocks.
\end{lemma}

\begin{proof}
We show an example with three clocks:\\
Let $x, y, z \in X$ be the three clocks, and $c_x = 2, c_y = 2, c_z = 2, v_1(x) = 0.4, v_1(y) = 0.4, v_1(z) = 0.8, v_2(x) = 0.1, v_2(y) = 0.1, v_2(z) = 0.95$\\
We show that $v_1 \cong v_2$:
\begin{enumerate}
\item $\lfloor v_1(x) \rfloor = 0, \lfloor v_2(x) \rfloor = 0 \implies \lfloor v_1(x) \rfloor = \lfloor v_2(x) \rfloor$\\
$\lfloor v_1(y) \rfloor = 0, \lfloor v_2(y) \rfloor = 0 \implies \lfloor v_1(y) \rfloor = \lfloor v_2(y) \rfloor$\\
$\lfloor v_1(z) \rfloor = 0, \lfloor v_2(z) \rfloor = 0 \implies \lfloor v_1(z) \rfloor = \lfloor v_2(z) \rfloor$
\item $fr(v_1(x)) = 0.4, fr(v_1(y)) = 0.4, fr(v_2(x)) = 0.1, fr(v_2(y)) = 0.1 \implies$\\
$fr(v_1(x)) \le fr(v_1(y))$ and $fr(v_2(x)) \le fr(v_2(y))$\\
Also true when we swap x and y.\\
$fr(v_1(x)) = 0.4, fr(v_1(z)) = 0.8, fr(v_2(x)) = 0.1, fr(v_2(z)) = 0.95 \implies$\\
$fr(v_1(x)) \le fr(v_1(z))$ and $fr(v_2(x)) \le fr(v_2(z))$\\
When swapped, both constraints fail as expected.\\
$fr(v_1(y)) = 0.4, fr(v_1(z)) = 0.8, fr(v_2(y)) = 0.1, fr(v_2(z)) = 0.95 \implies$\\
$fr(v_1(y)) \le fr(v_1(z))$ and $fr(v_2(x)) \le fr(v_2(z))$\\
When swapped, both constraints fail as expected.
\item None have fractional components of zero.
\item $fr(v_1(x)) = 0.4, fr(v_1(y)) = 0.4, fr(v_2(x)) = 0.1, fr(v_2(y)) = 0.1 \implies$\\
$fr(v_1(x)) + fr(v_1(y)) < 1$ and $fr(v_2(x)) + fr(v_2(y)) < 1$\\
Also true when we swap x and y.\\
$fr(v_1(x)) = 0.4, fr(v_1(z)) = 0.8, fr(v_2(x)) = 0.1, fr(v_2(z)) = 0.95 \implies$\\
$fr(v_1(x)) + fr(v_1(z)) > 1$ and $fr(v_2(x)) + fr(v_2(z)) > 1$\\
Also true when we swap x and z.\\
$fr(v_1(y)) = 0.4, fr(v_1(z)) = 0.8, fr(v_2(y)) = 0.1, fr(v_2(z)) = 0.95 \implies$\\
$fr(v_1(y)) + fr(v_1(z)) > 1$ and $fr(v_2(y)) + fr(v_2(z)) > 1$\\
Also true when we swap y and z.
\end{enumerate}
$\therefore v_1 \cong v_2$,
but if we choose $t_1 = 0.1 \implies (v_1 + t_1)(x) = 0.5, (v_1 + t_1)(y) = 0.5, (v_1 + t_1)(z) = 0.9$:
\begin{enumerate}
\item $\lfloor (v_2+t_2)(x) \rfloor = 0, \lfloor (v_2+t_2)(y) \rfloor = 0, \lfloor (v_2+t_2)(z) \rfloor = 0 \implies t_2 \le 0.9, t_2 \le 0.9, t_2 \le 0.05$\\
$t_2 < 0.05$
\item No fractional ordering will change unless there is an integral component change, so this has the same restriction of $t_2 < 0.05$.
\item $t_2 < 0.05$ for the same reason as above.
\item $fr((v_1 + t_1)(x)) + fr((v_1 + t_1)(y)) = 1 \implies fr((v_2 + t_2)(x)) + fr((v_2 + t_2)(y)) = 1$\\
$fr((v_1 + t_1)(x)) + fr((v_1 + t_1)(z)) > 1 \implies fr((v_2 + t_2)(x)) + fr((v_2 + t_2)(y)) > 1$\\
$fr((v_1 + t_1)(y)) + fr((v_1 + t_1)(z)) > 1 \implies fr((v_2 + t_2)(x)) + fr((v_2 + t_2)(y)) > 1$\\
$\implies t_2 = 0.4, t_2 < 0.475, t_2 < 0.475$
\end{enumerate}
We have two conflicting requirements where $t_2 < 0.05$ but $t_2 = 0.4$.
This means the interval for $z$ and the ordering on the fractional components must change before the $x+y$ sum constraint can be fulfilled.

\noindent
$\therefore$ There is no $t_2$ such that $v_1 + t_1 \cong v_2 + t_2$.
\end{proof}

We now extend our equivalence relation $\cong$ over clock valuations to an equivalence relation over the state space of timed automata by taking the cross product of the locations, clock valuations and the rate of change map, and require that equivalent states have the same location, clock region and rate of change map.\\
So let $s, s' \in S$, $v, v' \in V$ and $d, d' \in D$, then $(s, [v], d) \cong (s', [v'], d') \iff s = s' \land [v] = [v'] \land d = d'$.\\
We will call this the {\it region graph}.

\begin{lemma}
For every timed action in the timed automata, there exists a corresponding transition in the region graph.
\end{lemma}
\begin{proof}
Let $v_1 \cong v_2 \land (s, v_1) \Rightarrow_{a} (s', v_1')$. For simplicity, we will drop the encoding of the rate of change map since they can also be encoded in $s$.\\
The transition $\langle s, a, \phi, \lambda, \mu, s' \rangle$ that changes the state from $(s, v_1)$ to $(s', v_1')$ corresponds to two transitions of the timed automaton:
\begin{itemize}
\item a delay transition $(s, v_1) \to_{t_1} (s, v_1 + t_1)$ for some $t_1 \ge 0$, followed by
\item an action transition $(s, v_1 + t_1) \to_{a} (s, v_1') : v_1 + t_1$ satisfies $\phi$ and $v_1' = (v_1 + t_1)[\lambda := 0, \mu := c_\mu - \mu]$.
\end{itemize}

\noindent
Delay transition:\\
$v_1 \cong v_2 \implies \forall t_1 \in \mathbb{R}^+ \exists t_2 \in \mathbb{R}^+ [v_1 + t_1 \cong v_2 + t_2]$\\
$\therefore \exists t_2 : (s, v_2) \to_{t_2} (s, v_2 + t_2)$ and $(s, v_1 + t_1) \cong (s, v_2 + t_2)$

\noindent
Action transition:\\
$(v_1 + t_1) \cong (v_2 + t_2) \land (v_1 + t_1)$ satisfies $\phi \implies (v_2 + t_2)$ satisfies $\phi$\\
$(v_1 + t_1) \cong (v_2 + t_2) \implies (v_1 + t_1)[\lambda := 0] \cong (v_2 + t_2)[\lambda := 0]$\\
$(v_1 + t_1)[\lambda := 0] \cong (v_2 + t_2)[\lambda := 0] \implies (v_1 + t_1)[\lambda := 0,\mu := c_\mu - \mu] \cong (v_2 + t_2)[\lambda := 0,\mu := c_\mu - \mu]$\\
$\therefore \exists v_2' = (v_2 + t_2)[\lambda := 0,\mu := c_\mu - \mu] : v_1' \cong v_2' \land (s, v_2) \Rightarrow_{a} (s', v_2')$

\noindent
$\therefore$ We have transitions in the region graph that correspond to transitions of the timed automaton, and the behaviour is consistent.
\end{proof}

Lastly, the initial states in the region graph have the form:\\
\centerline{$(s_i, [v], d)$, where $s_i \in S_i$, and $\forall x \in X, v(x) = 0 \land d(x) = 1$}

The final states in the region graph have the form:\\
\centerline{$(s_f, [v], d)$, where $s \in S_f$}

\begin{lemma}
The constructed region graph will accept exactly the set of words equivalent to the words accepted by the corresponding timed automaton.
\end{lemma}

The proof for this will be the same as the proof given for the standard timed automata~\cite{ad94}.

\begin{theorem}
The language emptiness problem on hourglass automata with two or fewer clocks is decidable.
\end{theorem}



\section{Stopping Time}
\begin{lemma}
With two hourglass clocks, the stopping of time does not cause decidability issues.
\end{lemma}
Hourglass clocks are never compared to one another, and they are bounded to a range.
This means we only need to consider the progression of time.

In the 2-clock region diagram shown in Figure~\ref{fig:regiontypes}, we can see that a horizontal or vertical time progression can be made to work by introducing an extra constraint:\\
$\forall x \in X$ such that $v(x) \le c_x$
$fr(v(x)) \le 0.5 \iff fr(v'(x)) \le 0.5$ and $fr(v(x)) \ge 0.5 \iff fr(v'(x)) \ge 0.5$.
This cuts the second and fourth region in Figure~\ref{fig:regiontypes} into two halves.
Proving that $\cong$ remains an equivalence relation with this extra condition requires few modifications to the original ones since this new condition is an extended version of condition 4, but we let clocks $x$ and $y$ to be the same clock.

The two properties to prove are:
\begin{enumerate}
\item $v_1 \cong v_2 \land t \in \mathbb{Z}^+ \implies v_1 + t \cong v_2 + t$
\item $v_1 \cong v_2 \implies \forall t_1 \in \mathbb{R}^+ \exists t_2 \in \mathbb{R}^+ [v_1 + t_1 \cong v_2 + t_2]$
\end{enumerate}

\begin{proof}(1) This is unaffected from the proof given where no clocks are stopped.
\end{proof}

\begin{proof}(2) We can see from Figure~\ref{fig:regiontypes} that when we cut the second and fourth regions, the region transitions are fixed when we allow time to progress. Note that this does not work for three clocks since when you have one stopped clock with one of the other two clocks having a higher fractional component and the last clock having a lower fractional component than the stopped clock, you cannot determine whether the lower will cross the stopped clock or the higher will reach the next integer first.
\end{proof}

\section{Conclusion and Future Work}
In this paper, we introduced the hourglass automata and its ability to let clocks go backwards within some bounded range.
To be able to reason and study this class of automata, we first introduced an extension to the standard timed automata which has a new update, $x := c - x$, where $x$ is a clock, $c \in \mathbb{N}_{\le c_x}$, and $c_x$ is the greatest integer constant that clock $x$ is compared against, and a rate of change map $D : X \to \{-1, -0, 0, 1\}$ is added to the state space.
We then showed that this extended timed automata is capable of expressing the operations of the hourglass automata by showing a translation of the flip operation on bounded clocks.
As a result, this allowed us to examine the language emptiness problem for this class of automata.

For the hourglass automata, we limited the clock update to the form $x := c_x - x$ since this was all that was necessary. 
From there, we prove the decidability of the language emptiness problem on this class of automata when two or fewer clocks are involved, using the same approach used for proving decidability of the standard timed automata~\cite{ad94}.
This result shows that for two hourglass clocks, we can reduce the system to a finite graph.

The next step would be to investigate the possibility of having more than two clocks.
The timed automata that we currently have defined is capable of more than the hourglass automata definition, and maintains a lot of information, which is not necessary, so it is possible that we could use another model to represent the hourglass automata such that language emptiness verification is decidable with more than two clocks.
Some important notes here would be removing all unnecessary constraints that are not the boundaries, and possibly introduce clock regions which are not square with edges on the integer components. This may be possible since we don't compare to integer constants other than the bounds.

The ability for hourglasses to be placed on their side, allowing the stopping of clocks, will also be an interesting area to investigate when more than two clocks are introduced.
In general, the ability to stop time allows timed automata to become as expressive as linear hybrid automata~\cite{cl00}, thus introducing undecidability issues, but there has been other classes of timed automata that have a limited ability to stop time like the interrupt timed automata~\cite{bhs12}. The interrupt timed automata would have some overlap with the hourglass automata in what it can express.

Additionally, the decidability proof for the update given in this paper can apply for any $c \in \mathbb{N}_{\le c_x}$ as long as $c - x \ge 0$, otherwise the clock values become invalid, so this extended timed automata could be developed further.


\end{document}